\documentclass[12pt, reqno]{amsart} 
\usepackage{amsmath, amsthm, amscd, amsfonts, amssymb, graphicx, color}
\usepackage[bookmarksnumbered, colorlinks, plainpages]{hyperref}

\newtheorem{theorem}{Theorem}[section]
\newtheorem{lemma}[theorem]{Lemma}
\newtheorem{proposition}[theorem]{Proposition}
\newtheorem{corollary}[theorem]{Corollary}
\theoremstyle{definition}
\newtheorem{definition}[theorem]{Definition}
\newtheorem{example}[theorem]{Example}

\theoremstyle{remark}
\newtheorem{remark}[theorem]{Remark}
\numberwithin{equation}{section}

\begin{document}
\title[Triple cyclic codes over $\mathbb{Z}_2$]{Triple cyclic codes over $\mathbb{Z}_2$}
\author[Mostafanasab]{Hojjat Mostafanasab}

\date{17 Sep 2015}

\begin{abstract}
Let $r,s,t$ be three positive integers and $\mathcal{C}$ be a binary linear code of lenght $r+s+t$.
We say that $\mathcal{C}$ is a {\it triple cyclic code of lenght $(r,s,t)$ over $\mathbb{Z}_2$} if the set of coordinates can be partitioned into three parts that any cyclic shift of the coordinates of the parts leaves invariant the code. These codes can be considered as
$\mathbb{Z}_2[x]$-submodules of 
$\frac{\mathbb{Z}_2[x]}{\langle x^r-1\rangle}\times\frac{\mathbb{Z}_2[x]}{\langle x^s-1\rangle}\times\frac{\mathbb{Z}_2[x]}{\langle x^t-1\rangle}$.
We give the minimal generating sets of this kind of codes. Also, we determine the relationship between the generators
of triple cyclic codes and their duals.
\end{abstract}

\maketitle

\section{Introduction}
Codes over finite rings have been studied since the early 1970s. 
Recently codes over rings have
generated a lot of interest after a breakthrough paper by Hammons et al. \cite{HK}.
Cyclic codes are amongst the most studied algebraic codes.
Their structure is well known over finite fields \cite{M}. 
\par
In \cite{BFT}, Borges et. al. studied the algebraic structures of $\mathbb{Z}_2$-double cyclic codes as $\mathbb{Z}_2[x]$-submodules of $\mathcal{R}_{r,s}=\frac{\mathbb{Z}_2[x]}{\langle x^r-1\rangle}\times\frac{\mathbb{Z}_2[x]}{\langle x^s-1\rangle}$.
They determined the generator polynomials of this family of codes and their duals. In fact, the double cyclic codes were generalized quasi-cyclic (GQC) codes with index 2 introduced in \cite{Siap} and studied deeply by many other researchers \cite{Cao1,Cao2,Esmaeili,Gao}. 
Also, Gao et. al. \cite{GSWF} investigated double cyclic codes over $\mathbb{Z}_4$.
\par
In Section 2, we give the definition and $\mathbb{Z}_2$-module structure of triple cyclic codes.
In Section 3, we determine the generator polynomials and minimal generating sets of triple cyclic codes. 
In Section 4, we investigate the relationship between the generators of triple cyclic codes and their duals.
\section{Triple cyclic codes over $\mathbb{Z}_2$}
In this paper, suppose that $r,s,t$ are three positive integers and $\mathcal{C}$ is a binary linear code of lenght $n=r+s+t$. This code can be partitioned into
three parts of $r,s$ and $t$ coordinates, respectively.

\begin{definition}
Let $r,s,t$ be positive integers and $\mathcal{C}$ a binary linear code of lenght  $n=r+s+t$.
We say that $\mathcal{C}$ is a {\it triple cyclic code of lenght $(r,s,t)$ over} $\mathbb{Z}_2$ if
{\small
$$c=(c_{1,0},c_{1,1},\dots,c_{1,r-2},c_{1,r-1}\mid c_{2,0},c_{2,1},\dots,c_{2,s-2},c_{2,s-1}\mid c_{3,0},c_{3,1},\dots,c_{3,t-2},c_{3,t-1})\in\mathcal{C}$$}
implies that 
{\small
$$\mathcal{T}(c)=(c_{1,r-1},c_{1,0},c_{1,1},\dots,c_{1,r-2}\mid c_{2,s-1},c_{2,0},c_{2,1},\dots,c_{2,s-2}\mid c_{3,t-1},c_{3,0},c_{3,1}\dots,c_{3,t-2})\in\mathcal{C}.$$}
\end{definition}

Let $\mathcal{C}$ be a triple cyclic code of lenght $(r,s,t)$ over $\mathbb{Z}_2$.
Let $\mathcal{C}_r$ be the canonical projection of $\mathcal{C}$ on the first $r$ coordinates,
$\mathcal{C}_s$ on the second $s$ coordinates
and $\mathcal{C}_t$ on the last $t$ coordinates. It is easy to see that $\mathcal{C}_r$, $\mathcal{C}_s$ and
$\mathcal{C}_t$ are binary cyclic codes of lenght $r,s$ and $t$, respectively. A triple cyclic code $\mathcal{C}$
is called {\it separable} if $\mathcal{C}=\mathcal{C}_r\times\mathcal{C}_s\times\mathcal{C}_t$.

Let $\mathcal{R}_{r,s,t}$ be the ring $\frac{\mathbb{Z}_2[x]}{\langle x^r-1\rangle}\times\frac{\mathbb{Z}_2[x]}{\langle x^s-1\rangle}\times\frac{\mathbb{Z}_2[x]}{\langle x^t-1\rangle}$. The map
$$\Psi:\mathbb{Z}_2^r\times\mathbb{Z}_2^s\times\mathbb{Z}_2^t\to\mathcal{R}_{r,s,t}$$
which maps
{\small 
$$(u_{1,0},u_{1,1},\dots,u_{1,r-1}\mid u_{2,0},u_{2,1},\dots,u_{2,s-1}\mid u_{3,0},u_{3,1},\dots,u_{3,t-1})$$}
to
{\small
$$(u_{1,0}+u_{1,1}x+\dots+u_{1,r-1}x^{r-1}\mid u_{2,0}+u_{2,1}x+\dots+u_{2,s-1}x^{s-1}\mid u_{3,0}+u_{3,1}x+\dots+u_{3,t-1}x^{t-1})$$}
is an isomorphism of $\mathbb{Z}_2$-modules. We denote the image of a vector ${u}\in\mathbb{Z}_2^r\times\mathbb{Z}_2^s\times\mathbb{Z}_2^t$ by $u(x)$. 

\begin{definition}
We define the multiplication $$*:\mathbb{Z}_2[x]\times\mathcal{R}_{r,s,t}\to\mathcal{R}_{r,s,t}$$ as
$$\lambda(x)*\big(u_1(x)\mid u_2(x)\mid u_3(x)\big)=\big(\lambda(x)u_1(x)\mid\lambda(x)u_2(x)\mid\lambda(x)u_3(x)\big),$$ where $\lambda(x)\in\mathbb{Z}_2[x]$ and
$\big(u_1(x)\mid u_2(x)\mid u_3(x)\big)\in\mathcal{R}_{r,s,t}$.
\end{definition}
The ring $\mathcal{R}_{r,s,t}$ with the external multiplication $*$ is a $\mathbb{Z}_2[x]$-module.

Let $\mathcal{C}$ be a binary linear code of lenght $n$ and let
$$c=(c_{1,0},c_{1,1},\dots,c_{1,r-1}\mid c_{2,0},c_{2,1},\dots,c_{2,s-1}\mid c_{3,0},c_{3,1},\dots,c_{3,t-1})$$
be a codeword in $\mathcal{C}$. Note that $x*c(x)$ is equal to
{\small 
\begin{eqnarray*}
(c_{1,r-1}+c_{1,0}x+\dots+c_{1,r-2}x^{r-1}\mid c_{2,s-1}+c_{2,0}x+\dots+c_{2,s-2}x^{s-1}\\
\mid c_{3,t-1}+c_{3,0}x+\dots+c_{3,t-2}x^{t-1})
\end{eqnarray*}}
in $\mathcal{R}_{r,s,t}$, which is the image of $$(c_{1,r-1},c_{1,0},\dots,c_{1,r-2},\mid c_{2,s-1},c_{2,0},\dots,c_{2,s-2},\mid c_{3,t-1},c_{3,0},\dots,c_{3,t-2})$$ under $\Psi$.
Therefore $\mathcal{C}$ is a triple cyclic code if whenever $c(x)\in\mathcal{C}$, then $x*c(x)\in\mathcal{C}$ in $\mathcal{R}_{r,s,t}$.

\section{Properties of triple cyclic codes over $\mathbb{Z}_2$}

For a linear code $\mathcal{C}$, the {\it minimum Hamming distance} $d(\mathcal{C})$ is defined by
$$d(\mathcal{C})={\rm min}\{{\rm wt}(c)\mid 0\neq c\in\mathcal{C}\}.$$
For a linear code $\mathcal{C}$ with parity-check matrix $H$, any $d(\mathcal{C})-1$ columns of $H$ are linearly independent and $H$ has $d(\mathcal{C})$ columns that are linearly dependent.
\begin{proposition}
Let $\mathcal{C}$ be a triple cyclic code of lenght $(r,s,t)$ over $\mathbb{Z}_2$. 
\begin{enumerate}
\item $d(\mathcal{C})\geq min\{d(\mathcal{C}_r),d(\mathcal{C}_s),d(\mathcal{C}_t)\}.$
\item If $\mathcal{C}$ is separable, then $d(\mathcal{C})=min\{d(\mathcal{C}_r),d(\mathcal{C}_s),d(\mathcal{C}_t)\}.$
\end{enumerate}
\end{proposition}
\begin{proof}
(1) There exists a nonzero codeword $(c_r\mid c_s\mid c_t)$ of minimum distance in $\mathcal{C}$ such that
$d(\mathcal{C})={\rm wt}((c_r\mid c_s\mid c_t))$. Without loss of generality we may assume that $c_r\neq0$.
Therefore
$$d(\mathcal{C})={\rm wt}((c_r\mid c_s\mid c_t))\geq {\rm wt}(c_r)\geq d(\mathcal{C}_r)\geq {\rm min}\{d(\mathcal{C}_r),d(\mathcal{C}_s),d(\mathcal{C}_t)\}.$$
(2) Suppose that  $\mathcal{C}$ is separable. Assume that 
${\rm min}\{d(\mathcal{C}_r),d(\mathcal{C}_s),d(\mathcal{C}_t)\}=d(\mathcal{C}_r)$.
Let $0\neq c_r\in\mathcal{C}_r$ be such that $d(\mathcal{C}_r)={\rm wt}(c_r)$. On the other hand
$(c_r\mid0\mid0)\in\mathcal{C}$. So
$$d(\mathcal{C})\leq{\rm wt}((c_r\mid 0\mid 0))={\rm wt}(c_r)=d(\mathcal{C}_r)={\rm min}\{d(\mathcal{C}_r),d(\mathcal{C}_s),d(\mathcal{C}_t)\}.$$
Hence, by part (1) the claim holds.
\end{proof}

We know that $\mathcal{R}_{r,s,t}$ is a Noetherian $\mathbb{Z}_2[x]$-module, and so a triple cyclic code $\mathcal{C}$ as a
$\mathbb{Z}_2[x]$-submodule of $\mathcal{R}_{r,s,t}$ is finitely generated.

\begin{theorem}\label{T1}
Let $\mathcal{C}$ be a triple cyclic code of lenght $(r,s,t)$ over $\mathbb{Z}_2$. Then
$$\mathcal{C}=\big\langle\big(F_1(x)\mid0\mid0\big),\big(0\mid F_2(x)\mid0\big),\big(G_1(x)\mid G_2(x)\mid G_3(x)\big)\big\rangle$$ where
$F_1(x),F_2(x),G_1(x),G_2(x),G_3(x)\in\mathbb{Z}_2[x]$ with $F_1(x)\mid x^r-1$, $F_2(x)\mid x^s-1$ and $G_3(x)\mid x^t-1$.
\end{theorem}
\begin{proof}
Let $\Phi:\mathcal{C}\to\frac{\mathbb{Z}_2[x]}{\langle x^t-1\rangle}$ be the canonical projection of $\mathbb{Z}_2[x]$-modules
defined by $\Phi((c_1(x)\mid c_2(x)\mid c_3(x)))=c_3(x)$. 
Since ${\rm Im}(\Phi)$ is an ideal of $\frac{\mathbb{Z}_2[x]}{\langle x^t-1\rangle}$, then there exists $G_3(x)\in\mathbb{Z}_2[x]$ with $G_3(x)\mid x^t-1$ such that ${\rm Im}(\Phi)=\langle G_3(x)\rangle$.
We know that $${\rm Ker}(\Phi)=\{(c_1(x)\mid c_2(x)\mid0)\in\mathcal{R}_{r,s,t}\mid(c_1(x)\mid c_2(x)\mid0)\in\mathcal{C}\}.$$ Define 
$$\mathcal{I}=\{(c_1(x)\mid c_2(x))\in\mathcal{R}_{r,s}\mid(c_1(x)\mid c_2(x)\mid0)\in{\rm Ker}(\Phi)\}.$$
It is easy to check that $\mathcal{I}$ is an ideal of $\mathcal{R}_{r,s}$. So, $\mathcal{I}=\mathcal{I}_1\times\mathcal{I}_2$
for some ideal $\mathcal{I}_1$ of $\frac{\mathbb{Z}_2[x]}{\langle x^r-1\rangle}$ and some ideal $\mathcal{I}_2$ of $\frac{\mathbb{Z}_2[x]}{\langle x^s-1\rangle}$. Again, there are $F_1(x),F_2(x)\in\mathbb{Z}_2[x]$ with $F_1(x)\mid x^r-1$, $F_2(x)\mid x^s-1$ such that $\mathcal{I}_1=\langle F_1(x)\rangle$ and $\mathcal{I}_2=\langle F_2(x)\rangle$. Therefore
$\mathcal{I}=\big\langle\big(F_1(x)\mid0\big),\big(0\mid F_2(x)\big)\big\rangle$. Now, we can easily see that
$${\rm Ker}(\Phi)=\big\langle\big(F_1(x)\mid0\mid0\big),\big(0\mid F_2(x)\mid0\big)\big\rangle.$$ On the other hand, by the first isomorphism theorem,
we have $$\frac{\mathcal{C}}{{\rm Ker}(\Phi)}\simeq\langle G_3(x)\rangle.$$
Let $\big(G_1(x)\mid G_2(x)\mid G_3(x)\big)\in\mathcal{C}$ be such that $\Phi((G_1(x)\mid G_2(x)\mid G_3(x)))=G_3(x)$. Consequently
$\mathcal{C}$ as a $\mathbb{Z}_2[x]$-submodule of $\mathcal{R}_{r,s,t}$ is generated by elements 
$\big(F_1(x)\mid0\mid0\big)$,$\big(0\mid F_2(x)\mid0\big)$ and $\big(G_1(x)\mid G_2(x)\mid G_3(x)\big)$.
\end{proof}

\begin{remark}
Notice that if in a triple cyclic code $\mathcal{C}=\big\langle\big(F_1(x)\mid0\mid0\big),\big(0\mid F_2(x)\mid0\big),\big(G_1(x)\mid G_2(x)\mid G_3(x)\big)\big\rangle$ we have $G_3(x)=0$, then we may consider $\mathcal{C}$ as a double cyclic code
which was investigated in \cite{BFT}.
\end{remark}

We recall that, the {\it reciprocal polynomial} $f^*(x)$ of a
polynomial $f(x)=a_0+a_1x+\dots+a_nx^n$ is the polynomial $f^*(x)=a_n+a_{n-1}x+\dots+a_0x^n=x^nf(\frac{1}{x})$.\\
Also, we denote the polynomial $\sum\limits_{i=0}^{n-1}x^i$ by $\theta_n(x)$.

\begin{proposition}
Let $f(x),g(x)$ be two polynomials in $\mathbb{Z}_2[x]$ with $deg(f(x))\geq deg(g(x))$. Then the following conditions hold:
\begin{enumerate}
\item $deg(f(x))=deg(f^*(x))$.
\item $(f^*)^*(x)=f(x)$.
\item $(fg)^*(x)=f^*(x)g^*(x)$.
\item $(f+g)^*(x)=f^*(x)+x^{deg(f(x))-deg(g(x))}g^*(x)$.
\item $g(x)\mid f(x)$ if and only if $g^*(x)\mid f^*(x)$.
\item $gcd(f(x),g(x))^*=gcd(f^*(x),g^*(x))$.
\item $lcm(f(x),g(x))^*=lcm(f^*(x),g^*(x))$.
\end{enumerate}
\end{proposition}
\begin{proof}
(1) and (2) are easy.\\
For (3) and (4) see Lemma 4.3 of \cite{GG}.\\
(5) By parts (2) and (3).\\
(6) Since $gcd(f(x),g(x))$ divides both $f(x),g(x)$, then by part (5) it follows that  $gcd(f(x),g(x))^*$ divides $f^*(x),g^*(x)$.
Hence $gcd(f(x),g(x))^*\mid gcd(f^*(x),g^*(x))$. On the other hand there are two polynomials
$u(x),v(x)\in\mathbb{Z}_2[x]$ such that $$gcd(f(x),g(x))=u(x)f(x)+v(x)g(x).$$
Without loss of generality we may assume that $deg(u(x)f(x))\geq\deg(v(x)g(x))$. Set $l=deg(u(x)f(x))-\deg(v(x)g(x))$.
Therefore $$gcd(f(x),g(x))^*=u^*(x)f^*(x)+x^{l}v^*(x)g^*(x),$$
by part (4). So
$gcd(f^*(x),g^*(x))\mid gcd(f(x),g(x))^*$. Consequently $gcd(f(x),g(x))^*=gcd(f^*(x),g^*(x))$.\\
(7) Use the equlity $lcm(f(x),g(x))gcd(f(x),g(x))=f(x)g(x)$ and parts (3), (6).
\end{proof}

\begin{lemma}
Let ${a}=(a_0,a_1,\dots,a_{n-1})$ and ${b}=(b_0,b_1,\dots,b_{n-1})$ be vectors in $\mathbb{Z}^n_2$ with associated
polynomials $a(x)$ and $b(x)$. Then ${a}$ is orthogonal to ${b}$ and all its cyclic shifts if and only if
$a(x)b^*(x)=0$ mod $(x^n-1)$. 
\end{lemma}

\begin{proof}
See Lemma 4.4.8 of \cite{HP}.
\end{proof}

\begin{corollary}
Let $\mathcal{C}$ be a binary cyclic code of lenght $n$ with the dual code
$\mathcal{C}^{\bot}$. Then
$$\mathcal{C}^{\bot}=\{{a}\in\mathbb{Z}_2^n\mid 
a(x)b^*(x)=0 \mbox{ mod } (x^n-1) \mbox{ for every } {b}\in\mathcal{C}\}.$$
\end{corollary}

{\bf From now on we assume that ${\bf m=lcm(r,s,t)}$.}
\begin{remark}\label{rem1}
Regarding Proposition 4.2 of \cite{BFT}  we have that 
$$x^m-1=\theta_\frac{m}{r}(x^r)(x^r-1)=\theta_\frac{m}{s}(x^s)(x^s-1)=\theta_\frac{m}{t}(x^t)(x^t-1).$$
\end{remark}
\begin{definition}
Let $u(x)=(u_1(x)\mid u_2(x)\mid u_3(x))$ and $v(x)=(v_1(x)\mid v_2(x)\mid v_3(x))$ be two elements of $\mathcal{R}_{r,s,t}$.
We define the map $$\circ:\mathcal{R}_{r,s,t}\times\mathcal{R}_{r,s,t}\to\frac{\mathbb{Z}_2[x]}{\langle x^m-1\rangle}$$
with 
\begin{eqnarray*}
\circ(u(x),v(x))&=&u_1(x)\theta_{\frac{m}{r}}(x^r)x^{m-1-deg(v_1(x))}v_1^*(x)+u_2(x)\theta_{\frac{m}{s}}(x^s)x^{m-1-deg(v_2(x))}v_2^*(x)\\
&+&u_3(x)\theta_{\frac{m}{t}}(x^t)x^{m-1-deg(v_3(x))}v_3^*(x)\hspace{.5cm} \mbox{mod} ~(x^m-1).
\end{eqnarray*}
The map $\circ$ is a bilinear map between $\mathbb{Z}_2[x]$-modules.
\end{definition}

\begin{proposition}\label{P6}
Let ${u},{v}$ be two elements of $\mathbb{Z}_2^r\times\mathbb{Z}_2^s\times\mathbb{Z}_2^t$. Then
$$u(x)\circ v(x)=0 ~\mbox{~mod~} (x^m-1)$$ if and only if ${u}$ is orthogonal to ${v}$ and all its shifts.
\end{proposition}
\begin{proof}
Consider the following representations for ${u},{v}$:
\begin{eqnarray*}
{u}&=&(u_{1,0},u_{1,1},\dots,u_{1,r-1}\mid u_{2,0},u_{2,1},\dots,u_{2,s-1}\mid u_{3,0},u_{3,1},\dots,u_{3,t-1}),\\
{v}&=&(v_{1,0},v_{1,1},\dots,v_{1,r-1}\mid v_{2,0},v_{2,1},\dots,v_{2,s-1}\mid v_{3,0},v_{3,1},\dots,v_{3,t-1}).
\end{eqnarray*}
Assume that 
$${v}^{(i)}=(v_{1,0-i},v_{1,1-i},\dots,v_{1,r-1-i}\mid v_{2,0-i},v_{2,1-i},\dots,v_{2,s-1-i}\mid v_{3,0-i},v_{3,1-i},\dots,v_{3,t-1-i})$$
is the $i$-th cyclic shift of ${v}$, where $0\leq i\leq m-1$.
Notice that ${u}\cdot {v}^{(i)}=0$ if and only if 
$$\sum\limits_{j=0}^{r-1}u_{1,j}v_{1,j-i}+\sum\limits_{k=0}^{s-1}u_{2,k}v_{2,k-i}+\sum\limits_{l=0}^{t-1}u_{3,l}v_{3,l-i}=0.$$
Set $S_i:=\sum\limits_{j=0}^{r-1}u_{1,j}v_{1,j-i}+\sum\limits_{k=0}^{s-1}u_{2,k}v_{2,k-i}+\sum\limits_{l=0}^{t-1}u_{3,l}v_{3,l-i}$.
Therefore 
\begin{eqnarray*}
u(x)\circ v(x)&=&\theta_{\frac{m}{r}}(x^r)\sum\limits_{h=0}^{r-1}\sum\limits_{j=0}^{r-1}u_{1,j}v_{1,j-h}x^{m-1-h}+\theta_{\frac{m}{s}}(x^s)\sum\limits_{p=0}^{s-1}\sum\limits_{k=0}^{s-1}u_{2,k}v_{2,k-p}x^{m-1-p}\\&+&\theta_{\frac{m}{t}}(x^t)\sum\limits_{q=0}^{t-1}\sum\limits_{l=0}^{t-1}u_{3,l}v_{3,l-q}x^{m-1-q}=\sum\limits_{i=0}^{m-1}S_ix^{m-1-i}
\hspace{1cm}\mbox{mod}~~(x^m-1).
\end{eqnarray*}
Consequently $u(x)\circ v(x)=0 ~\mbox{~mod~} (x^m-1)$ if and only if $S_i=0$ for every $0\leq i\leq m-1$.
\end{proof}

\begin{proposition}\label{P5}
Let $u(x)=(u_1(x)\mid u_2(x)\mid u_3(x))$ and $v(x)=(v_1(x)\mid v_2(x)\mid v_3(x))$ be two elements of $\mathcal{R}_{r,s,t}$ such that
$u_2(x)=0$ or $v_2(x)=0$, and $u_3(x)=0$ or $v_3(x)=0$. Then $u(x)\circ v(x)=0$~ mod $(x^m-1)$ if and only if $u_1(x)v_1^*(x)=0$~ mod $(x^r-1)$.
\end{proposition}
\begin{proof}
($\Rightarrow$) Similar to that of \cite[Lemma 4.5]{BFT}.\\
($\Leftarrow$) Suppose that $u_1(x)v_1^*(x)=0$~ mod $(x^r-1)$. Then, there exists $\lambda(x)\in\mathbb{Z}_2[x]$ such that
$u_1(x)v_1^*(x)=\lambda(x)(x^r-1)$, and so 
$$u(x)\circ v(x)=u_1(x)\theta_{\frac{m}{r}}(x^r)x^{m-1-deg(v_1(x))}v_1^*(x)=x^{m-1-deg(v_1(x))}\lambda(x)\theta_{\frac{m}{r}}(x^r)(x^r-1).$$
Therefore, by Remark \ref{rem1} we have that 
$$u(x)\circ v(x)=x^{m-1-deg(v_1(x))}\lambda(x)(x^m-1),$$
which is $0$ mod $(x^m-1)$.
\end{proof}

Similar to Proposition \ref{P5} we can state the next two propositions.
\begin{proposition}
Let $u(x)=(u_1(x)\mid u_2(x)\mid u_3(x))$ and $v(x)=(v_1(x)\mid v_2(x)\mid v_3(x))$ be two elements of $\mathcal{R}_{r,s,t}$ such that
$u_1(x)=0$ or $v_1(x)=0$, and $u_3(x)=0$ or $v_3(x)=0$. Then $u(x)\circ v(x)=0$~ mod $(x^m-1)$ if and only if $u_2(x)v_2^*(x)=0$~ mod $(x^s-1)$.
\end{proposition}

\begin{proposition}
Let $u(x)=(u_1(x)\mid u_2(x)\mid u_3(x))$ and $v(x)=(v_1(x)\mid v_2(x)\mid v_3(x))$ be two elements of $\mathcal{R}_{r,s,t}$ such that
$u_1(x)=0$ or $v_1(x)=0$, and $u_2(x)=0$ or $v_2(x)=0$. Then $u(x)\circ v(x)=0$~ mod $(x^m-1)$ if and only if $u_3(x)v_3^*(x)=0$~ mod $(x^t-1)$.
\end{proposition}

\begin{proposition}\label{P2}
Let $\mathcal{C}=\big\langle(F_1(x)\mid0\mid0),(0\mid F_2(x)\mid0),(G_1(x)\mid G_2(x)\mid G_3(x))\big\rangle$ be a triple cyclic code of lenght $(r,s,t)$ over $\mathbb{Z}_2$. Then 
\begin{enumerate}
\item $F_1(x)\mid\frac{x^t-1}{G_3(x)}G_1(x)$ and $F_2(x)\mid\frac{x^t-1}{G_3(x)}G_2(x)$.
\item $F_1(x)F_2(x)\mid\frac{x^t-1}{G_3(x)}gcd(F_1(x)F_2(x),F_1(x)G_2(x),F_2(x)G_1(x))$.
\item $\mathcal{C}_r=\big\langle gcd(F_1(x),G_1(x))\big\rangle$, $\mathcal{C}_s=\big\langle gcd(F_2(x),G_2(x))\big\rangle$ and $\mathcal{C}_t=\big\langle G_3(x)\big\rangle$.
\item $(\mathcal{C}_r)^{\bot}=\big\langle\frac{x^r-1}{gcd(F_1^*(x),G_1^*(x))}\big\rangle$, 
$(\mathcal{C}_s)^{\bot}=\big\langle\frac{x^s-1}{gcd(F_2^*(x),G_2^*(x))}\big\rangle$ and $(\mathcal{C}_t)^{\bot}=\big\langle \frac{x^t-1}{G_3^*(x)}\big\rangle$.
\end{enumerate}
\end{proposition}
\begin{proof}
(1) Consider the projection homomorphism of $\mathbb{Z}_2[x]$-modules
$$\Phi:\mathcal{C}\to\frac{\mathbb{Z}_2[x]}{\langle x^t-1\rangle}$$
$$\hspace{.9mm}(c_1(x)\mid c_2(x)\mid c_3(x))\mapsto c_3(x).$$\\
In view of the proof of Theorem \ref{T1}, ${\rm Ker}(\Phi)=\langle(F_1(x)\mid0\mid0),(0\mid F_2(x)\mid0)\rangle$. On the other hand, we have that
{\small
$$\frac{x^t-1}{G_3(x)}*(G_1(x)\mid G_2(x)\mid G_3(x))
=(\frac{x^t-1}{G_3(x)}G_1(x)\mid\frac{x^t-1}{G_3(x)}G_2(x)\mid0)$$
$$\hspace{7cm}\subseteq{\rm Ker}(\pi)=\langle(F_1(x)\mid0\mid0),(0\mid F_2(x)\mid0)\rangle.$$}
Consequently $F_1(x)\mid\frac{x^t-1}{G_3(x)}G_1(x)$ and $F_2(x)\mid\frac{x^t-1}{G_3(x)}G_2(x)$.\\
(2) By part (1).\\
(3) We show that $\mathcal{C}_r=\big\langle gcd(F_1(x),G_1(x))\big\rangle$.
Let $u(x)\in\mathcal{C}_r$. Then there exist
$v(x)\in\frac{\mathbb{Z}_2[x]}{\langle x^s-1\rangle}$ and $w(x)\in\frac{\mathbb{Z}_2[x]}{\langle x^t-1\rangle}$
such that
 $(u(x)\mid v(x)\mid w(x))\in\mathcal{C}$. Thus there are
$\lambda(x),\mu(x),\nu(x)\in\mathbb{Z}_2[x]$ such that 
{\small
$$(u(x)\mid v(x)\mid w(x))=\lambda(x)(F_1(x)\mid0\mid0)+\mu(x)(0\mid F_2(x)\mid0)+\nu(x)(G_1(x),G_2(x),G_3(x)).$$}
Hence $u(x)=\lambda(x)F_1(x)+\nu(x)G_1(x)$. Then
$gcd(F_1(x),G_1(x))$ divides $u(x)$. So $u(x)\in\big\langle gcd(F_1(x),G_1(x))\big\rangle$. Thus
$\mathcal{C}_r\subseteq\big\langle gcd(F_1(x),G_1(x))\big\rangle$. On the other hand there exist 
$\eta(x),\gamma(x)\in\mathbb{Z}_2[x]$ such that $gcd(F_1(x),G_1(x))=\eta(x)F_1(x)+\gamma(x)G_1(x)$.
Then 
\begin{eqnarray*}
\big(gcd(F_1(x),G_1(x))\mid \gamma G_2(x)\mid\gamma G_3(x)\big)&=&\eta(x)(F_1(x)\mid0\mid0)\\
&&+\gamma(x)(G_1(x)\mid G_2(x)\mid G_3(x))\in\mathcal{C}.
\end{eqnarray*}
Therefore $gcd(F_1(x),G_1(x))\in\mathcal{C}_r$, which shows that $\mathcal{C}_r=\big\langle gcd(F_1(x),G_1(x))\big\rangle$.\\
(4) By part (3) and \cite[Theorem 4.2.7]{HP}.
\end{proof}

As a direct consequence of parts (3),(4) of Proposition \ref{P2} we have the following result.
\begin{corollary}
Let $\mathcal{C}=\big\langle(F_1(x)\mid0\mid0),(0\mid F_2(x)\mid0),(G_1(x)\mid G_2(x)\mid G_3(x))\big\rangle$ be a triple cyclic code of lenght $(r,s,t)$ over $\mathbb{Z}_2$. Then
$$|\mathcal{C}_r|=2^{r-\deg(\gcd(F_1(x),G_1(x)))}, ~~|\mathcal{C}_s|=2^{s-\deg(\gcd(F_2(x),G_2(x)))},~~|\mathcal{C}_t|=2^{t-\deg(G_3(x))},$$
$$|(\mathcal{C}_r)^\perp|=2^{\deg(\gcd(F_1(x),G_1(x)))},~~ |(\mathcal{C}_s)^\perp|=2^{\deg(\gcd(F_2(x),G_2(x)))},~~|(\mathcal{C}_t)^\perp|=2^{\deg(G_3(x))}.$$
\end{corollary}

Let $S$ be a subset of $\mathcal{R}_{r,s,t}$. The $\mathbb{Z}_2$-submodule of 
$\mathcal{R}_{r,s,t}$ generated by $S$ is denoted by $\langle S\rangle_{\mathbb{Z}_2}$.
\begin{theorem}\label{T2}
Let $\mathcal{C}=\big\langle(F_1(x)\mid0\mid0),(0\mid F_2(x)\mid0),(G_1(x)\mid G_2(x)\mid G_3(x))\big\rangle$ be a triple cyclic code of lenght $(r,s,t)$ over $\mathbb{Z}_2$. Define the sets
\begin{eqnarray*}
S_1&=&\bigcup\limits_{i=0}^{r-deg(F_1(x))-1}\{x^{i}*(F_1(x)\mid0\mid0)\},\\
S_2&=&\bigcup\limits_{i=0}^{s-deg(F_2(x))-1}\{x^{i}*(0\mid F_2(x)\mid0)\},\\
S_3&=&\bigcup\limits_{i=0}^{t-deg(G_3(x))-1}\{x^{i}*(G_1(x)\mid G_2(x)\mid G_3(x))\}.
\end{eqnarray*}
Then the following conditions hold:
\begin{enumerate}
\item $\langle S_1\rangle_{\mathbb{Z}_2}=\big\langle(F_1(x)\mid0\mid0)\big\rangle$.
\item $\langle S_2\rangle_{\mathbb{Z}_2}=\big\langle(0\mid F_2(x)\mid0)\big\rangle$.
\item $\langle S_1\cup S_2\cup S_3\rangle_{\mathbb{Z}_2}\supseteq\big\langle(G_1(x)\mid G_2(x)\mid G_3(x))\big\rangle$.
\item $S_1\cup S_2\cup S_3$ forms a minimal generating set for $\mathcal{C}$ as a $\mathbb{Z}_2$-submodule of $\mathcal{R}_{r,s,t}$.
\item $\mid\mathcal{C}\mid=2^d$ where $d=r+s+t-deg(F_1(x))-deg(F_2(x))-deg(G_3(x))$.
\end{enumerate}
\end{theorem}
\begin{proof}
(1) It is obvious that $\langle S_1\rangle_{\mathbb{Z}_2}\subseteq\big\langle(F_1(x)\mid0\mid0)\big\rangle$.
Let $p_1(x)\in\mathbb{Z}_2[x]$. We show that $p_1(x)*(F_1(x)\mid0\mid0)\in\langle S_1\rangle_{\mathbb{Z}_2}$.
If $deg(p_1(x))\leq r-deg(F_1(x))-1$, then we are done. Otherwise, there exist polynomials $q_1(x),r_1(x)\in\mathbb{Z}_2[x]$
such that $p_1(x)=\frac{x^r-1}{F_1(x)}q_1(x)+r_1(x)$ where $r_1(x)=0$ or $deg(r_1(x))\leq r-deg(F_1(x))-1$. Therefore
\begin{eqnarray*}
p_1(x)*(F_1(x)\mid0\mid0)&=&\frac{x^r-1}{F_1(x)}q_1(x)*(F_1(x)\mid0\mid0)+r_1(x)*(F_1(x)\mid0\mid0)\\
&=&q_1(x)*(\frac{x^r-1}{F_1(x)}F_1(x)\mid0\mid0)+r_1(x)*(F_1(x)\mid0\mid0)\\
&=&r_1(x)*(F_1(x)\mid0\mid0)\in\langle S_1\rangle_{\mathbb{Z}_2}.
\end{eqnarray*}
So $\big\langle(F_1(x)\mid0\mid0)\big\rangle\subseteq\langle S_1\rangle_{\mathbb{Z}_2}$ and the equality holds.\\
(2) Similar to the proof of part (1).\\
(3) Get a polynomial $p_2(x)\in\mathbb{Z}_2[x]$. We prove that $p_2(x)*(G_1(x)\mid G_2(x)\mid G_3(x))\in\langle S_1\cup S_2\cup S_3\rangle_{\mathbb{Z}_2}$. If $deg(p_2(x))\leq t-deg(G_3(x))-1$, then $p_2(x)*(G_1(x)\mid G_2(x)\mid G_3(x))\in\langle S_3\rangle_{\mathbb{Z}_2}$.
Otherwise, there exist $q_2(x),r_2(x)\in\mathbb{Z}_2[x]$
such that $p_2(x)=\frac{x^t-1}{G_3(x)}q_2(x)+r_2(x)$ where $r_2(x)=0$ or $deg(r_2(x))\leq t-deg(G_3(x))-1$. Hence
{\small
\begin{eqnarray*}
p_2(x)*(G_1(x)\mid G_2(x)\mid G_3(x))&=&\frac{x^t-1}{G_3(x)}q_2(x)*(G_1(x)\mid G_2(x)\mid G_3(x))\\&+&r_2(x)*(G_1(x)\mid G_2(x)\mid G_3(x))\\
&=&q_2(x)*(\frac{x^t-1}{G_3(x)}G_1(x)\mid\frac{x^t-1}{G_3(x)}G_2(x)\mid0)\\&+&r_2(x)*(G_1(x)\mid G_2(x)\mid G_3(x)).
\end{eqnarray*}}
Clearly $r_2(x)*(G_1(x)\mid G_2(x)\mid G_3(x))\in\langle S_3\rangle_{\mathbb{Z}_2}$. By Proposition \ref{P2}(1),
$F_1(x)\mid\frac{x^t-1}{G_3(x)}G_1(x)$ and $F_2(x)\mid\frac{x^t-1}{G_3(x)}G_2(x)$. So, parts (1) and (2) imply that
$$q_2(x)*(\frac{x^t-1}{G_3(x)}G_1(x)\mid\frac{x^t-1}{G_3(x)}G_2(x)\mid0)\in\langle S_1\cup S_2\rangle_{\mathbb{Z}_2}.$$
Consequently the claim holds.\\
(4) By the previous parts.\\
(5) By part (4).
\end{proof}

\begin{corollary}
Let $\mathcal{C}$ be a triple cyclic code of lenght $(r,s,t)$ over $\mathbb{Z}_2$.
\begin{enumerate}
\item If $\mathcal{C}=\big\langle(F_1(x)\mid0\mid0)\big\rangle$ where
$F_1(x)\in\mathbb{Z}_2[x]$ with $F_1(x)\mid x^r-1$, then every
codeword $c(x)$ of $\mathcal{C}$ is in the form of 
$c(x)=p(x)*(F_1(x)\mid0\mid0)$
where $p(x)$ is a polynomial in $\mathbb{Z}_2[x]$ with $deg(p(x))=r-deg(F_1(x))-1$.
\item If $\mathcal{C}=\big\langle(0\mid F_2(x)\mid0)\big\rangle$ where
$F_2(x)\in\mathbb{Z}_2[x]$ with $F_2(x)\mid x^s-1$, then every
codeword $c(x)$ of $\mathcal{C}$ is in the form of $c(x)=p(x)*(0\mid F_2(x)\mid0)$
where $p(x)$ is a polynomial in $\mathbb{Z}_2[x]$ with $deg(p(x))=s-deg(F_2(x))-1$.
\item If $\mathcal{C}=\big\langle(G_1(x)\mid G_2(x)\mid G_3(x))\big\rangle$ where
$G_1(x),G_2(x),G_3(x)\in\mathbb{Z}_2[x]$ with  $G_3(x)\mid x^t-1$, then every
codeword $c(x)$ of $\mathcal{C}$ is in the form of 
$c(x)=p(x)*(G_1(x)\mid G_2(x)\mid G_3(x))$
where  $p(x)$ is a polynomial in $\mathbb{Z}_2[x]$ with $deg(p(x))=t-deg(G_3(x))-1$.
\item If $\mathcal{C}=\big\langle(F_1(x)\mid0\mid0),(0\mid F_2(x)\mid0)\big\rangle$ where
$F_1(x),F_2(x)\in\mathbb{Z}_2[x]$ with $F_1(x)\mid x^r-1$, $F_2(x)\mid x^s-1$, then every
codeword $c(x)$ of $\mathcal{C}$ is in the form of 
$$c(x)=p_1(x)*(F_1(x)\mid0\mid0)+p_2(x)*(0\mid F_2(x)\mid0)$$
where $p_1(x)$ and $p_2(x)$ are polynomials in $\mathbb{Z}_2[x]$ with 
$$deg(p_1(x))=r-deg(F_1(x))-1 \mbox{ and } deg(p_2(x))=s-deg(F_2(x))-1.$$
\item If $\mathcal{C}=\big\langle(F_1(x)\mid0\mid0),(G_1(x)\mid G_2(x)\mid G_3(x))\big\rangle$ where
$F_1(x),G_1(x),G_2(x),\\G_3(x)\in\mathbb{Z}_2[x]$ with $F_1(x)\mid x^r-1$ and $G_3(x)\mid x^t-1$, then every
codeword $c(x)$ of $\mathcal{C}$ is in the form of 
$$c(x)=p_1(x)*(F_1(x)\mid0\mid0)+p_2(x)*(G_1(x)\mid G_2(x)\mid G_3(x))$$
where $p_1(x)$ and $p_2(x)$ are polynomials in $\mathbb{Z}_2[x]$ with 
$$deg(p_1(x))=r-deg(F_1(x))-1 \mbox{ and } deg(p_2(x))=t-deg(G_3(x))-1.$$
\item If $\mathcal{C}=\big\langle(0\mid F_2(x)\mid0),(G_1(x)\mid G_2(x)\mid G_3(x))\big\rangle$ where
$F_2(x),G_1(x),G_2(x),\\G_3(x)\in\mathbb{Z}_2[x]$ with $F_2(x)\mid x^s-1$ and $G_3(x)\mid x^t-1$, then every
codeword $c(x)$ of $\mathcal{C}$ is in the form of 
$$c(x)=p_1(x)*(0\mid F_2(x)\mid0)+p_2(x)*(G_1(x)\mid G_2(x)\mid G_3(x))$$
where $p_1(x)$ and $p_2(x)$ are polynomials in $\mathbb{Z}_2[x]$ with 
$$deg(p_1(x))=s-deg(F_2(x))-1 \mbox{ and } deg(p_2(x))=t-deg(G_3(x))-1.$$
\item If $\mathcal{C}=\big\langle(F_1(x)\mid0\mid0),(0\mid F_2(x)\mid0),(G_1(x)\mid G_2(x)\mid G_3(x))\big\rangle$ where
$F_1(x),F_2(x),G_1(x),G_2(x),G_3(x)\in\mathbb{Z}_2[x]$ with $F_1(x)\mid x^r-1$, $F_2(x)\mid x^s-1$ and $G_3(x)\mid x^t-1$, then every
codeword $c(x)$ of $\mathcal{C}$ is in the form of 
{\small
$$\hspace{1cm}c(x)=p_1(x)*(F_1(x)\mid0\mid0)+p_2(x)*(0\mid F_2(x)\mid0)+p_3(x)*(G_1(x)\mid G_2(x)\mid G_3(x))$$}
where $p_1(x)$,$p_2(x)$ and $p_3(x)$ are polynomials in $\mathbb{Z}_2[x]$ with $deg(p_1(x))=r-deg(F_1(x))-1$, $deg(p_2(x))=s-deg(F_2(x))-1$ and $deg(p_3(x))=t-deg(G_3(x))-1$.
\end{enumerate}
\end{corollary}

\begin{proposition}\label{P9}
Let $\mathcal{C}=\big\langle(F_1(x)\mid0\mid0),(0\mid F_2(x)\mid0),(G_1(x)\mid G_2(x)\mid G_3(x))\big\rangle$ be a triple cyclic code of lenght $(r,s,t)$ over $\mathbb{Z}_2$. 
Then $F_1(x)\mid G_1(x)$ if and only if
$\mathcal{C}=\big\langle(F_1(x)\mid0\mid0),(0\mid F_2(x)\mid0),(0\mid G_2(x)\mid G_3(x))\big\rangle$, i.e, we may assume that
$G_1(x)=0$.
\end{proposition}
\begin{proof}
The ``if'' part is evident. \\
Suppose that $F_1(x)\mid G_1(x)$. 
Then, there exists a polynomial $\lambda(x)$ in $\mathbb{Z}_2[x]$ such that $G_1(x)=\lambda(x)F_1(x)$.
Set $$\mathcal{C}^{\prime}=\big\langle(F_1(x)\mid0\mid0),(0\mid F_2(x)\mid0),(0\mid G_2(x)\mid G_3(x))\big\rangle.$$
Notice that $$(0\mid G_2(x)\mid G_3(x))=\lambda(x)(F_1(x)\mid0\mid0)+(G_1(x)\mid G_2(x)\mid G_3(x)).$$
Hence $\mathcal{C}^{\prime}\subseteq\mathcal{C}$.
On the other hand
$$(G_1(x)\mid G_2(x)\mid G_3(x))=\lambda(x)(F_1(x)\mid0\mid0)+(0\mid G_2(x)\mid G_3(x)).$$
So $\mathcal{C}\subseteq\mathcal{C}^{\prime}$.
\end{proof}

Similar to the previous proposition we have the next result.
\begin{proposition}\label{P10}
Let $\mathcal{C}=\big\langle(F_1(x)\mid0\mid0),(0\mid F_2(x)\mid0),(G_1(x)\mid G_2(x)\mid G_3(x))\big\rangle$ be a triple cyclic code of lenght $(r,s,t)$ over $\mathbb{Z}_2$. 
Then $F_2(x)\mid G_2(x)$ if and only if
$\mathcal{C}=\big\langle(F_1(x)\mid0\mid0),(0\mid F_2(x)\mid0),(G_1(x)\mid 0\mid G_3(x))\big\rangle$, i.e, we may assume that
$G_2(x)=0$.
\end{proposition}

\begin{proposition}\label{P11}
Let $\mathcal{C}=\big\langle(F_1(x)\mid0\mid0),(0\mid F_2(x)\mid0),(G_1(x)\mid G_2(x)\mid G_3(x))\big\rangle$ be a triple cyclic code of lenght $(r,s,t)$ over $\mathbb{Z}_2$. 
The following conditions are equivalent:
\begin{enumerate}
\item $\mathcal{C}$ is separable;
\item $F_1(x)\mid G_1(x)$ and $F_2(x)\mid G_2(x)$;
\item $\mathcal{C}_r=\langle F_1(x)\rangle$ and $\mathcal{C}_s=\langle F_2(x)\rangle$;
\item $\mathcal{C}=\big\langle(F_1(x)\mid0\mid0),(0\mid F_2(x)\mid0),(0\mid0\mid G_3(x))\big\rangle$, i.e, we may assume that
$G_1(x)=0$ and $G_2(x)=0$.
\end{enumerate}
\end{proposition}
\begin{proof}
(1)$\Rightarrow$(2) Assume that $\mathcal{C}$ is separable. Then 
$$\mathcal{C}=\mathcal{C}_r\times\mathcal{C}_s\times\mathcal{C}_t=\langle gcd(F_1(x),G_1(x))\rangle\times\langle gcd(F_2(x),G_2(x))\rangle\times\langle G_3(x)\rangle,$$
by Proposition \ref{P2}(3).
Since $(gcd(F_1(x),G_1(x))\mid0\mid0)\in\mathcal{C}$, then we can deduce that $gcd(F_1(x),G_1(x))=\lambda(x)F_1(x)$
for some $\lambda(x)\in\mathbb{Z}_2[x]$.
Therefore  $F_1(x)\mid G_1(x)$. Also, it is easy to verify that $F_2(x)\mid G_2(x)$.\\
(2)$\Leftrightarrow$(3) is straightforward.\\
(2)$\Rightarrow$(4) Suppose that $F_1(x)\mid G_1(x)$ and $F_2(x)\mid G_2(x)$. 
Then, there exist two polynomials $\lambda_1(x),\lambda_2(x)$ in $\mathbb{Z}_2[x]$ such that $G_1(x)=\lambda_1(x)F_1(x)$
and $G_2(x)=\lambda_2(x)F_2(x)$.
So, by the equality
 $$(0\mid0\mid G_3(x))=\lambda_1(x)(F_1(x)\mid0\mid0)+\lambda_2(x)(0\mid F_2(x)\mid0)+(G_1(x)\mid G_2(x)\mid G_3(x)).$$
the result follows.\\
(4)$\Rightarrow$(1) Assume that
$\mathcal{C}=\big\langle(F_1(x)\mid0\mid0),(0\mid F_2(x)\mid0),(0\mid0\mid G_3(x))\big\rangle$.
Hence
$\mathcal{C}=\langle F_1(x)\rangle\times\langle F_2(x)\rangle\times\langle G_3(x)\rangle=\mathcal{C}_r\times\mathcal{C}_s\times\mathcal{C}_t.$
Then $\mathcal{C}$ is separable.
\end{proof}

\begin{proposition}\label{P1}
Let $\mathcal{C}=\big\langle(F_1(x)\mid0\mid0),(0\mid F_2(x)\mid0),(G_1(x)\mid G_2(x)\mid G_3(x))\big\rangle$ be a triple cyclic code of lenght $(r,s,t)$ over $\mathbb{Z}_2$. 
The following conditions hold:
\begin{enumerate}
\item It can be assumed that $deg(G_1(x))\leq deg(F_1(x))$ and $deg(G_2(x))\leq deg(F_2(x))$.
\item $\mathcal{C}=\big\langle(F_1(x)\mid0\mid0),(0\mid F_2(x)\mid0),(F_1(x)+G_1(x)\mid F_2(x)+G_2(x)\mid G_3(x))\big\rangle$.
\item If $G_3(x)=0$, then $$\mathcal{C}\subseteq\big\langle(gcd(F_1(x),G_1(x))\mid0\mid0),(0\mid gcd(F_2(x),G_2(x))\mid0)\big\rangle.$$
\item If $G_1(x)=G_3(x)=0$, then $\mathcal{C}=\big\langle(F_1(x)\mid0\mid0),(0\mid gcd(F_2(x),G_2(x))\mid0)\big\rangle$.
\item If $G_2(x)=G_3(x)=0$, then $\mathcal{C}=\big\langle(gcd(F_1(x),G_1(x))\mid0\mid0),(0\mid F_2(x)\mid0)\big\rangle$.
\end{enumerate}
\end{proposition}
\begin{proof}
(1) Suppose that $deg(G_1(x))>deg(F_1(x))$ and set 
$$\mathcal{C}^{\prime}=\big\langle(F_1(x)\mid0\mid0),(0\mid F_2(x)\mid0),(G_1(x)+x^{{l}}F_1(x)\mid G_2(x)\mid G_3(x))\big\rangle$$ where ${l}=deg(G_1(x))-deg(F_1(x))$. Notice that $deg(G_1(x)+x^{{l}}F_1(x))<deg(G_1(x))$.
Since $$(G_1(x)+x^{{l}}F_1(x)\mid G_2(x)\mid G_3(x))=x^{{l}}*(F_1(x)\mid0\mid0)+(G_1(x)\mid G_2(x)\mid G_3(x))\in\mathcal{C},$$ then
$\mathcal{C}^{\prime}\subseteq\mathcal{C}$. On the other hand,
$$(G_1(x)\mid G_2(x)\mid G_3(x))=(G_1(x)+x^{{l}}F_1(x)\mid G_2(x)\mid G_3(x))-x^{{l}}*(F_1(x)\mid0\mid0).$$
Hence $\mathcal{C}^{\prime}=\mathcal{C}$. So we would be able to reduce the degree of $G_1(x)$ in $\mathcal{C}$ to reach the claim.
An argument like above can be stated for $deg(G_2(x))\leq deg(F_2(x))$.\\
(2),(3),(4) and (5) are easy.
\end{proof}

\begin{example}
Let $\mathcal{C}=\big\langle(1+x^2\mid0\mid0),(0\mid x+x^5\mid0),(x^3+x^4+x^5\mid x^2+x^6\mid G_3(x))\big\rangle$ be a triple cyclic code over $\mathbb{Z}_2$. Regarding the proof of Proposition \ref{P1},
{\small 
$$\begin{array}{ll}
\mathcal{C}$=$\big\langle(1+x^2\mid0\mid0),(0\mid x+x^5\mid0),(x^3+x^4+x^5+x^{3}(1+x^2)\mid x^2+x^6+x(x+x^5)\mid G_3(x))\big\rangle\\
\hspace{2mm}$=$\big\langle(1+x^2\mid0\mid0),(0\mid x+x^5\mid0),(x^4\mid 0\mid G_3(x))\big\rangle\\
\hspace{2mm}$=$\big\langle(1+x^2\mid0\mid0),(0\mid x+x^5\mid0),(x^4+x^2(1+x^2)\mid 0\mid G_3(x))\big\rangle\\
\hspace{2mm}$=$\big\langle(1+x^2\mid0\mid0),(0\mid x+x^5\mid0),(x^2\mid 0\mid G_3(x))\big\rangle.
\end{array}$$}
\end{example}

\section{Dual codes of triple cyclic codes over $\mathbb{Z}_2$}

\begin{proposition}
$\mathcal{C}$ is a triple cyclic code of lenght $(r,s,t)$ over $\mathbb{Z}_2$ if and only if
$\mathcal{C}^{\bot}$ is a triple cyclic code of lenght $(r,s,t)$ over $\mathbb{Z}_2$.
Moreover, 
$$\mathcal{C}^{\bot}=\{{u}\in\mathbb{Z}_2^r\times\mathbb{Z}_2^s\times\mathbb{Z}_2^t\mid 
u(x)\circ c(x)=0 \mbox{ mod } (x^m-1) \mbox{ for every } {c}\in\mathcal{C}\}.$$
\end{proposition}
\begin{proof}
($\Rightarrow$) Suppose that $\mathcal{C}$ is a triple cyclic code of lenght $(r,s,t)$ over $\mathbb{Z}_2$. Assume that
$${c}^{\prime}=(c^{\prime}_{1,0},c^{\prime}_{1,1},\dots,c^{\prime}_{1,r-1}\mid c^{\prime}_{2,0},c_{2,1},\dots,c^{\prime}_{2,s-1}\mid c^{\prime}_{3,0},c^{\prime}_{3,1},\dots,c^{\prime}_{3,t-1})$$
is a codeword of $\mathcal{C}^\bot$. It is sufficient to show that $\mathcal{T}({c}^{\prime})\in\mathcal{C}^{\bot}$.
Let
$${c}=(c_{1,0},c_{1,1},\dots,c_{1,r-1}\mid c_{2,0},c_{2,1},\dots,c_{2,s-1}\mid c_{3,0},c_{3,1},\dots,c_{3,t-1})$$
be an arbitrary codeword of $\mathcal{C}$. 
Set $m:={\rm lcm}(r,s,t)$. Obviously we have  $\mathcal{T}^m({c})={c}$. Hence 
{\small
$$\mathcal{T}^{m-1}({c})=(c_{1,1},c_{1,2},\dots,c_{1,r-1},c_{1,0}\mid c_{2,1},c_{2,2},\dots,c_{2,s-1},c_{2,0}\mid c_{3,1},c_{3,2},\dots,c_{3,t-1},c_{3,0})\in\mathcal{C}.$$}
\hspace{-1mm}Therefore ${c}^{\prime}\cdot \mathcal{T}^{m-1}({c})=0$, because ${c}^{\prime}\in\mathcal{C}^\bot$. So
{\small
\begin{eqnarray*}
0&=&{c}^{\prime}\cdot \mathcal{T}^{m-1}({c})\\
&=&c^{\prime}_{1,0}c_{1,1}+\dots+c^{\prime}_{1,r-2}c_{1,r-1}+c^{\prime}_{1,r-1}c_{1,0}
+c^{\prime}_{2,0}c_{2,1}+\dots+c^{\prime}_{2,s-2}c_{2,s-1}+c^{\prime}_{2,s-1}c_{2,0}\\
&+&c^{\prime}_{3,0}c_{3,1}+\dots+c^{\prime}_{3,t-2}c_{3,t-1}+c^{\prime}_{3,t-1}c_{3,0}\\
&=&c_{1,0}c^{\prime}_{1,r-1}+c_{1,1}c^{\prime}_{1,0}+\dots+c_{1,r-1}c^{\prime}_{1,r-2}+
c_{2,0}c^{\prime}_{2,s-1}+c_{2,1}c^{\prime}_{2,0}+\dots+c_{2,s-1}c^{\prime}_{2,s-2}\\
&+&c_{3,0}c^{\prime}_{3,t-1}+c_{3,1}c^{\prime}_{3,0}+\dots+c_{3,t-1}c^{\prime}_{3,t-2}\\
&=&{c}\cdot \mathcal{T}({c}^{\prime}).
\end{eqnarray*}}
\hspace{-1mm}Thus $\mathcal{T}({c}^{\prime})\in\mathcal{C}^{\bot}$. Consequently $\mathcal{C}^{\bot}$ is a triple cyclic code of lenght $(r,s,t)$ over $\mathbb{Z}_2$.\\
($\Leftarrow$) By the fact that for every linear code $\mathcal{C}$, $(\mathcal{C}^{\bot})^\bot=\mathcal{C}$.\\
For the second statement use Proposition \ref{P6}.
\end{proof}

\begin{proposition}\label{PP3}
Let $\mathcal{C}$ be a triple cyclic code of lenght $(r,s,t)$ over $\mathbb{Z}_2$. Then 
\begin{enumerate}
\item $(\mathcal{C}_r)^\bot=\{{a}\in\mathbb{Z}_2^r|({a}\mid0\mid0)\in\mathcal{C}^\bot\}=
\{a(x)\in\frac{\mathbb{Z}_2[x]}{\langle x^r-1\rangle}|(a(x)\mid0\mid0)\in\mathcal{C}^\bot\}$, and so $(\mathcal{C}_r)^\bot\subseteq(\mathcal{C}^\bot)_r$.
\item $(\mathcal{C}_s)^\bot=\{{b}\in\mathbb{Z}_2^s|(0\mid{b}\mid0)\in\mathcal{C}^\bot\}
=\{b(x)\in\frac{\mathbb{Z}_2[x]}{\langle x^s-1\rangle}|(0\mid b(x)\mid0)\in\mathcal{C}^\bot\}$, and so $(\mathcal{C}_s)^\bot\subseteq(\mathcal{C}^\bot)_s$.
\item $(\mathcal{C}_t)^\bot=\{{c}\in\mathbb{Z}_2^t|(0\mid0\mid{c})\in\mathcal{C}^\bot\}
=\{c(x)\in\frac{\mathbb{Z}_2[x]}{\langle x^t-1\rangle}|(0\mid0\mid c(x))\in\mathcal{C}^\bot\}$, and so $(\mathcal{C}_t)^\bot\subseteq(\mathcal{C}^\bot)_t$.

\end{enumerate}
\end{proposition}
\begin{proof}
Straightforward.
\end{proof}

\begin{proposition}\label{PP2}
Let $\mathcal{C}=\big\langle(F_1(x)\mid0\mid0),(0\mid F_2(x)\mid0),(G_1(x)\mid G_2(x)\mid G_3(x))\big\rangle$ be a triple cyclic code of lenght $(r,s,t)$ over $\mathbb{Z}_2$. Then 
\begin{enumerate}
\item $(\frac{(x^r-1)^2}{F_1^*(x)G_1^*(x)}\mid0\mid0),(0\mid\frac{(x^s-1)^2}{F_2^*(x)G_2^*(x)}\mid0),(0\mid0\mid\frac{x^t-1}{G_3^*(x)})\in\mathcal{C}^{\bot}$.
\item $(\mathcal{C}_r)^\bot\subseteq(\mathcal{C}^\bot)_r\subseteq\langle\frac{x^r-1}{F_1^*(x)}\rangle$ and $(\mathcal{C}_s)^\bot\subseteq(\mathcal{C}^\bot)_s\subseteq\langle\frac{x^s-1}{F_2^*(x)}\rangle$.
\item If $F_1(x)\mid G_1(x)$, then $(\mathcal{C}^\bot)_r=(\mathcal{C}_r)^\bot=\langle\frac{x^r-1}{F_1^*(x)}\rangle$ and so 
$|(\mathcal{C}^\perp)_r|=2^{deg(F_1(x))}$.
\item If $F_2(x)\mid G_2(x)$, then $(\mathcal{C}^\bot)_s=(\mathcal{C}_s)^\bot=\langle\frac{x^s-1}{F_2^*(x)}\rangle$ and so $|(\mathcal{C}^\perp)_s|=2^{deg(F_2(x))}$.
\item If $F_1(x)\mid G_1(x)$ and $F_2(x)\mid G_2(x)$, then $\mathcal{C}^\bot=\langle\frac{x^r-1}{F_1^*(x)}\rangle\times\langle\frac{x^s-1}{F_2^*(x)}\rangle\times\langle\frac{x^t-1}{G_3^*(x)}\rangle$
and  $|\mathcal{C}^\perp|=2^{deg(F_1(x))+deg(F_2(x))+deg(G_3(x))}$.
Moreover, $(\mathcal{C}^\bot)_t=(\mathcal{C}_t)^\bot=\langle\frac{x^t-1}{G_3^*(x)}\rangle$ and so $|(\mathcal{C}^\perp)_t|=2^{deg(G_3(x))}$. 
\end{enumerate}
\end{proposition}
\begin{proof}
(1) We only show that $(\frac{(x^r-1)^2}{F_1^*(x)G_1^*(x)}\mid0\mid0)\in\mathcal{C}^{\bot}$.
Notice that $\frac{(x^r-1)^2}{F_1^*(x)G_1^*(x)}F_1^*(x)=(x^r-1)\frac{(x^r-1)}{G_1^*(x)}=0$~ mod $(x^r-1)$. Now,
Proposition \ref{P5} implies that $(\frac{(x^r-1)^2}{F_1^*(x)G_1^*(x)}\mid0\mid0)\circ(F_1(x)\mid0\mid0)=0$~ mod $(x^m-1)$.
Similarly we can show that $(\frac{(x^r-1)^2}{F_1^*(x)G_1^*(x)}\mid0\mid0)\circ(G_1(x)\mid G_2(x)\mid G_3(x))=0$~ mod $(x^m-1)$.
Clearly $(\frac{(x^r-1)^2}{F_1^*(x)G_1^*(x)}\mid0\mid0)\circ(0\mid F_2(x)\mid0)=0$~ mod $(x^m-1)$. So the result follows.\\
(2) We prove that $(\mathcal{C}^\bot)_r\subseteq\langle\frac{x^r-1}{F_1^*(x)}\rangle$. Let $f(x)\in(\mathcal{C}^\bot)_r$. Then there exist $g(x)\in\frac{\mathbb{Z}_2[x]}{\langle x^s-1\rangle}$ and 
$h(x)\in\frac{\mathbb{Z}_2[x]}{\langle x^t-1\rangle}$ such that $(f(x)\mid g(x)\mid h(x))\in\mathcal{C}^\bot$. Hence
$(f(x)\mid g(x)\mid h(x))\circ(F_1(x)\mid0\mid0)=0$~ mod $(x^m-1)$. So $f(x)F_1^*(x)=0$~mod $(x^r-1)$, see Proposition \ref{P5}.
Thus, there exists a $\lambda(x)\in\mathbb{Z}_2[x]$ such that $f(x)=\lambda(x)\frac{x^r-1}{F_1^*(x)}$. Consequently
$f(x)\subseteq\langle\frac{x^r-1}{F_1^*(x)}\rangle$ and we are done. Similarly it can be shown that $(\mathcal{C}^\bot)_s\subseteq\langle\frac{x^s-1}{F_2^*(x)}\rangle$.\\
(3) Suppose that $F_1(x)\mid G_1(x)$, then by Proposition \ref{P9} we may assume that
$G_1(x)=0$. Hence $(\frac{x^r-1}{F_1^*(x)}\mid0\mid0)\in\mathcal{C}^{\bot}$, and so
 $\langle\frac{x^r-1}{F_1^*(x)}\rangle\subseteq(\mathcal{C}_r)^\bot$, by Proposition \ref{PP3}(1).
 Now, by part (2) we have that $(\mathcal{C}^\bot)_r=(\mathcal{C}_r)^\bot=\langle\frac{x^r-1}{F_1^*(x)}\rangle$.\\
(4) An argument similar to the proof of part (3) can be stated.\\
(5) Use Proposition \ref{P11}.
\end{proof}

\begin{proposition}\label{P3}
Let $\mathcal{C}=\big\langle(F_1(x)\mid0\mid0),(0\mid F_2(x)\mid0),(G_1(x)\mid G_2(x)\mid G_3(x))\big\rangle$ be a triple cyclic code of lenght $(r,s,t)$ over $\mathbb{Z}_2$ with the dual code $\mathcal{C}^{\bot}=\big\langle(\widehat{F_1}(x)\mid0\mid0),(0\mid \widehat{F_2}(x)\mid0),(\widehat{G_1}(x)\mid \widehat{G_2}(x)\mid \widehat{G_3}(x))\big\rangle$. Then
\begin{enumerate}
\item $(\mathcal{C}_r)^\bot=\langle\widehat{F}_1(x)\rangle$, $\widehat{F_1}(x)=\frac{x^r-1}{gcd(F_1^*(x),G_1^*(x))}$ and $$deg(\widehat{F_1}(x))=r-deg(gcd(F_1(x),G_1(x))).$$
\item $(\mathcal{C}_s)^\bot=\langle\widehat{F}_2(x)\rangle$, $\widehat{F_2}(x)=\frac{x^s-1}{gcd(F_2^*(x),G_2^*(x))}$ and $$deg(\widehat{F_2}(x))=s-deg(gcd(F_2(x),G_2(x))).$$
\item $\widehat{F_1}(x)\mid\frac{(x^r-1)^2}{F_1^*(x)G_1^*(x)}$ and $\widehat{F_2}(x)\mid\frac{(x^s-1)^2}{F_2^*(x)G_2^*(x)}$.
\item $(\mathcal{C}_t)^\bot\subseteq\langle\widehat{G}_3(x)\rangle$ and so $\widehat{G_3}(x)\mid\frac{x^t-1}{G_3^*(x)}$.
\item If $F_1(x)\mid G_1(x)$ and $F_2(x)\mid G_2(x)$, then $\widehat{G_3}(x)=\frac{x^t-1}{G_3^*(x)}$ and so 
$deg(\widehat{G_3}(x))=t-deg(G_3(x))$.
\item $\widehat{G_1}(x)=\nu(x)\frac{(x^r-1)}{F_1^*(x)}$ for some $\nu(x)\in\mathbb{Z}_2[x]$ with 
$$deg(\nu(x))\leq deg(F_1(x))-deg(gcd(F_1(x),G_1(x))).$$
\item $\widehat{G_2}(x)=\rho(x)\frac{(x^s-1)}{F_2^*(x)}$ for some $\rho(x)\in\mathbb{Z}_2[x]$ with 
$$deg(\rho(x))\leq deg(F_2(x))-deg(gcd(F_2(x),G_2(x))).$$
\item $\widehat{G_3}(x)=\sigma(x)\frac{(x^t-1)gcd(F_1^*(x)F_2^*(x),F_1^*(x)G_2^*(x),F_2^*(x)G_1^*(x))}{F_1^*(x)F_2^*(x)G_3^*(x)}$ for some $\sigma(x)\in\mathbb{Z}_2[x]$.
\end{enumerate}
\end{proposition}
\begin{proof}
(1) Let $a(x)\in(\mathcal{C}_r)^\bot$. Then by Proposition \ref{PP3}(1), $(a(x)\mid0\mid0)\in\mathcal{C}^\bot$.
Hence, clearly $a(x)\in\langle\widehat{F_1}(x)\rangle$. So $(\mathcal{C}_r)^\bot\subseteq\langle\widehat{F_1}(x)\rangle$.
Since $(\widehat{F_1}(x)\mid0\mid0)\in\mathcal{C}^{\bot}$, again by Proposition \ref{PP3}(1), $\widehat{F_1}(x)\in(\mathcal{C}_r)^\bot$.
Thus $(\mathcal{C}_r)^\bot=\langle\widehat{F_1}(x)\rangle$. Now, see part (4) of Proposition \ref{P2}\\
(2) Similar to part (1).\\
(3) By Proposition \ref{PP3}, Proposition \ref{PP2}(1) and the previous parts.\\
(4) Similar to part (1).\\
(5) Notice that $(\mathcal{C}^\bot)_t=\langle\widehat{G_3}(x)\rangle$. Now, use Proposition \ref{PP2}(5).\\
(6) Since $(\widehat{G_1}(x)\mid \widehat{G_2}(x)\mid \widehat{G_3}(x))\in\mathcal{C}^{\bot}$, then from
$$(\widehat{G_1}(x)\mid \widehat{G_2}(x)\mid \widehat{G_3}(x))\circ(F_1(x)\mid0\mid0)=0\hspace{.5cm}\mbox{mod}~(x^m-1)$$
it follows that $\widehat{G_1}(x)F_1^*(x)=0$ mod $(x^r-1)$. Hence there exists a $\nu(x)\in\mathbb{Z}_2[x]$ such that
$\widehat{G_1}(x)=\nu(x)\frac{(x^r-1)}{F_1^*(x)}$. For the second claim, use part (1) and Proposition \ref{P1}(1).\\
(7) Similart to part (6).\\
(8) Set ${g}(x):=gcd(F_1(x)F_2(x),F_1(x)G_2(x),F_2(x)G_1(x))$. Notice that
{\small
\begin{eqnarray*}
\frac{F_1(x)F_2(x)}{{g}(x)}(G_1(x)\mid G_2(x)\mid G_3(x))&-&\frac{F_2(x)G_1(x)}{{g}(x)}(F_1(x)\mid0\mid0)\\
&-&\frac{F_1(x)G_2(x)}{{g}(x)}(0\mid F_2(x)\mid0)\\&=&(0\mid 0\mid\frac{F_1(x)F_2(x)G_3(x)}{{g}(x)})
\in\mathcal{C}.
\end{eqnarray*}}
Since $(\widehat{G_1}(x)\mid \widehat{G_2}(x)\mid \widehat{G_3}(x))\in\mathcal{C}^{\bot}$, then
$$(\widehat{G_1}(x)\mid \widehat{G_2}(x)\mid \widehat{G_3}(x))\circ(0\mid 0\mid\frac{F_1(x)F_2(x)G_3(x)}{{g}(x)})=0\hspace{0.5cm}\mbox{mod}~(x^m-1).$$
Hence $\widehat{G_3}(x)\frac{F_1^*(x)F_2^*(x)G_3^*(x)}{gcd(F_1^*(x)F_2^*(x),F_1^*(x)G_2^*(x),F_2^*(x)G_1^*(x))}=0$ mod~$(x^t-1)$.
Consequently there exists a $\sigma(x)\in\mathbb{Z}_2[x]$ such that  $\widehat{G_3}(x)=\sigma(x)\frac{(x^t-1)gcd(F_1^*(x)F_2^*(x),F_1^*(x)G_2^*(x),F_2^*(x)G_1^*(x))}{F_1^*(x)F_2^*(x)G_3^*(x)}$.
\end{proof}

The proof of the next proposition is similar to that of Proposition 3.3 of \cite{BFT}. 
\begin{proposition}\label{P4}
Let $\mathcal{C}=\big\langle(F_1(x)\mid0\mid0),(0\mid F_2(x)\mid0),(G_1(x)\mid 0\mid G_3(x))\big\rangle$ be a triple cyclic code of lenght $(r,s,t)$ over $\mathbb{Z}_2$. Then $\mathcal{C}$ is permutation equivalent to a code with the generator matrix in the form of
$$G =
 \left(\begin{array}{ccc|cc|ccc}
  I_{r-\deg(F_1(x))} & A_1 & A_2 &0&0& 0 & 0 & 0 \\
  0&0&0&I_{s-deg(F_2(x))}&C&0&0&0\\
  0 & B_{\kappa} & B &0&0& D_1 & I_\kappa &  0 \\
  0 & 0 & 0 &0&0& D_2  & D_3 & I_{t-\deg(G_3(x))-\kappa} 
 \end{array}\right),
$$
in which $B_\kappa$ is a full rank square matrix of size $\kappa\times\kappa$, where $\kappa=\deg(F_1(x))-\deg(\gcd(F_1(x),G_1(x)))$.
\end{proposition}

\begin{proposition}\label{PPP4}
Let $\mathcal{C}=\big\langle(F_1(x)\mid0\mid0),(0\mid F_2(x)\mid0),(G_1(x)\mid 0\mid G_3(x))\big\rangle$ be a triple cyclic code of lenght $(r,s,t)$ over $\mathbb{Z}_2$. Then
$$|(\mathcal{C}^\perp)_r|=2^{deg(F_1(x))}, ~~|(\mathcal{C}^\perp)_s|=2^{deg(F_2(x))},~~|(\mathcal{C}^\perp)_t|=2^{deg(G_3(x))+\kappa},$$
where $\kappa=deg(F_1(x))-deg(\gcd(F_1(x),G_1(x)))$.
\end{proposition}

\begin{proof}
The values of the cardinalities can be obtained from the projections on the first $r$, second $s$ and the last $t$ coordinates of the parity-check matrix of
$\mathcal{C}$.
\end{proof}

\begin{proposition}
Let $\mathcal{C}=\big\langle(F_1(x)\mid0\mid0),(0\mid F_2(x)\mid0),(G_1(x)\mid 0\mid G_3(x))\big\rangle$ be a triple cyclic code of lenght $(r,s,t)$ over $\mathbb{Z}_2$ with the dual code $\mathcal{C}^{\bot}=\big\langle(\widehat{F_1}(x)\mid0\mid0),(0\mid \widehat{F_2}(x)\mid0),(\widehat{G_1}(x)\mid \widehat{G_2}(x)\mid \widehat{G_3}(x))\big\rangle$. Then 
$$deg(\widehat{G_3}(x))=t-deg(G_3(x))-deg(F_1(x))+deg(\gcd(F_1(x),G_1(x)))$$ and 
$\widehat{G_3}(x)=\frac{(x^t-1)\gcd(F_1^*(x),G_1^*(x))}{F_1^*(x)G_3^*(x)}$.
\end{proposition}
\begin{proof}
First, note that $(\mathcal{C}^\perp)_t=\langle\widehat{G_3}(x)\rangle$. So $|(\mathcal{C}^\perp)_t|=2^{t-deg(\widehat{G_3}(x))}$.
On the other hand $|(\mathcal{C}^\perp)_t|=2^{deg(G_3(x))+deg(F_1(x))-deg(\gcd(F_1(x),G_1(x)))}$, Proposition \ref{PPP4}. Therefore
$$deg(\widehat{G_3}(x))=t-deg(G_3(x))-deg(F_1(x))+deg(\gcd(F_1(x),G_1(x))).$$ 
Since $G_2(x)=0$, then $\widehat{G_3}(x)=\sigma(x)\frac{(x^t-1)\gcd(F_1^*(x),G_1^*(x))}{F_1^*(x)G_3^*(x)}$ for some $\sigma(x)\in\mathbb{Z}_2[x]$, by Proposition \ref{P3}(8). It is easy to see that $deg(\sigma(x))=0$, and so $\sigma(x)=1$.
\end{proof}

\begin{proposition}
Let $\mathcal{C}=\big\langle(F_1(x)\mid0\mid0),(0\mid F_2(x)\mid0),(G_1(x)\mid 0\mid G_3(x))\big\rangle$ be a triple cyclic code of lenght $(r,s,t)$ over $\mathbb{Z}_2$ with the dual code 
$$\mathcal{C}^{\bot}=\big\langle(\widehat{F_1}(x)\mid0\mid0),(0\mid \widehat{F_2}(x)\mid0),(\widehat{G_1}(x)\mid \widehat{G_2}(x)\mid \widehat{G_3}(x))\big\rangle.$$ Let $\widehat{G_1}(x)=\nu(x)\frac{(x^r-1)}{F_1^*(x)}$ and
$\zeta(x)=\frac{G_1(x)}{gcd(F_1(x),G_1(x))}$. Then
\begin{enumerate}
\item $\nu(x) x^{m-deg(G_1(x))-1}\zeta^*(x)+x^{m-deg(G_3(x))-1}=0$ mod $\frac{F_1^*(x)}{gcd(F_1^*(x),G_1^*(x))}$.
\item $\nu(x)=x^{m-deg(G_3(x))+deg(G_1(x))}(\zeta^*(x))^{-1}$ mod $\frac{F_1^*(x)}{gcd(F_1^*(x),G_1^*(x))}$.
\end{enumerate}
\end{proposition}

\begin{proof}
The proof is similar to that of Proposition 4.18 and Corollary 4.19 of \cite{BFT}.
\end{proof}

\vspace{5mm} \noindent \footnotesize 
\begin{minipage}[b]{10cm}
Hojjat Mostafanasab \\
Department of Mathematics and Applications, \\ 
University of Mohaghegh Ardabili, \\ 
P. O. Box 179, Ardabil, Iran. \\
Email: h.mostafanasab@gmail.com, \hspace{1mm} h.mostafanasab@uma.ac.ir
\end{minipage}\\


\begin{thebibliography}{99}

\bibitem{BFT} J. Borges, C. Fern\'{a}ndez-C\'{o}rdoba and R. Ten-Valls, \emph{$\mathbb{Z}_2$-double cyclic codes}, arXiv preprint, arXiv: 1410.5604v1.

\bibitem{Cao2} Y. Cao, \emph{Generalized quasi-cyclic codes over Galois rings: structural properties and enumeration}, Appl. Algebra Eng. Commun. Comput. \textbf{22}, 219--233 (2011).

\bibitem{Cao1} Y. Cao, \emph{Structural properties and enumeration of $1$-generator generalized quasi-cyclic codes}, Des. Codes Crypto. \textbf{60}, 67--79 (2011).

\bibitem{Esmaeili} M. Esmaeili and S. Yari, \emph{Generalized quasi-cyclic codes: structural properties and codes construction}. Appl. Algebra Eng. Commun. Comput. \textbf{20}, 159-173 (2009).

\bibitem{Gao} J. Gao, F.-W. Fu, L. Shen and W. Ren, \emph{Some results on Generalized quasi-cyclic codes over $\mathbb{F}_q+u\mathbb{F}_q$}. IEICE Trans. Fund. \textbf{97}, 1005-1011 (2014).

\bibitem{GSWF} J. Gao, M. Shi, T. Wu and F. W. Fu, \emph{On double cyclic codes over $\mathbb{Z}_4$}, arXiv preprint, arXiv: 1501.01360v1.

\bibitem{GG} K. Guenda and T. A. Gulliver, {\it Construction of cyclic codes over $\mathbb{F}_2+u\mathbb{F}_2$ for DNA
computing}. { Appl. Algebra Eng. Commun. Comput.} \textbf{24}(6) 445--459 (2013).

\bibitem{HK} A. Hammons, P. V. Kumar, A. R. Calderbank, N. J. A. Slone and P. Sol´e, {\it The $\mathbb{Z}_4$ linearity
of Kerdock, Preparata, Goethals and related codes}, {IEEE Trans. Inf. Theory,} \textbf{40}(4) 301--319 (1994).

\bibitem{HP} W. C. Huffman and V. Pless, {\it Fundamentals of Error-Correcting Codes}, Cambridge University Press, Cambridge, 2003.

\bibitem{M} F. J. MacWilliams and N. J. A. Sloane, {\it The theory of error correcting codes}, North Holland, 1977.

\bibitem{Siap} I. Siap and N. Kulhan, \emph{The structure of generalized quasi-cyclic codes}, Appl. Math. E-Notes. \textbf{5}, 24--30 (2005).
\end{thebibliography}
\end{document}